\documentclass{amsart}
\usepackage{times, stmaryrd, verbatim, tikz-cd}
\usepackage{ amssymb, mathrsfs} 

\usepackage{color}

\newtheorem{theorem}{Theorem}[section]
\newtheorem{lemma}[theorem]{Lemma}
\newtheorem{proposition}[theorem]{Proposition}
\newtheorem{corollary}{Corollary}
\newtheorem*{claim}{Claim}
\newtheorem{Theorem}{Theorem}
\newtheorem*{Theorem*}{Theorem}

\theoremstyle{definition}
\newtheorem{definition}{Definition}[section]
\newtheorem{Definition}{Definition} 
\newtheorem{example}{Example}
\newtheorem*{Definition*}{Definition}
\newtheorem*{example*}{Example}

\theoremstyle{remark}
\newtheorem{remark}{Remark}

\numberwithin{equation}{section}

\usepackage{hyperref}
\begin{document}

\title{Tulczyjew's Approach for Particles in Gauge Fields}
\author{Guowu Meng}

\address{Department of Mathematics, Hong Kong University of Science and Technology, Clear Water Bay, Kowloon, Hong Kong}

\email{mameng@ust.hk}

\thanks{The author is supported by Hong Kong University of Science and Technology under DAG S09/10.SC02.}
  

\date{May 1, 2014}


\maketitle
\begin{abstract}
In mid-1970s W. M. Tulczyjew discovered an approach to classical mechanics which brings the Hamiltonian formalism and the Lagrangian formalism under a common geometric roof: the dynamics of a particle with configuration space $X$ is determined by a Lagrangian submanifold $D$ of $TT^*X$ (the total tangent space of $T^*X$), and the description of $D$ by its Hamiltonian $H$: $T^*X\to \mathbb R$ (resp. its Lagrangian $L$: $TX\to\mathbb R$) yields the Hamilton (resp. Euler-Lagrange) equation. 

It is reported here that Tulczyjew's approach also works for the dynamics of (charged) particles in gauge fields, in which the role of the total cotangent space $T^*X$ is played by Sternberg phase spaces. In particular, it is shown that, for a particle in a gauge field, the equation of motion can be locally presented as the Euler-Lagrange equation for a Lagrangian which is the sum of the ordinary Lagrangian $L(q, \dot q)$, the Lorentz term, and an {\em extra new term} which vanishes whenever the gauge group is abelian.  A charge quantization condition is also derived, generalizing Dirac's charge quantization condition from $\mathrm{U}(1)$ gauge group to any compact connected gauge group.
\end{abstract}

\tableofcontents
 \section{Introduction}
Tulczyjew's unified approach \cite{Tulczyjew1974} to Lagrangian and Hamiltonian descriptions of particle dynamics, though not well-known outside a small circle of mathematicians and physicists, is quite appealing to geometry-oriented minds. In this approach, the Legendre transformation takes a specially simple interpretation, and systems with singular Lagrangians or subject to constraints appear naturally. An advantage of this geometric approach is its flexibility in the sense that it can be easily adapted to different settings. As one more demonstration of its flexibility, in this article it is shown that Tulczyjew's approach also works for particle dynamics in which a charged particle moves in the presence of an external gauge field, either abelian or non-abelian. 

The incorporation of a gauge field into the classical particle dynamics is a nontrivial business when the gauge group is non-abelian.  It seems that, in the setting of symplectic geometry, this was initially done by S. Sternberg \cite{Sternberg77}.  Some further elaborations, especially the one on the relationship of Sternberg's symplectic approach with the earlier Poisson approach of S. K. Wong \cite{Wong70}, came from A. Weinstein \cite{Weinstein78} and R. Montgomery \cite{Montgomery84}. 

For readers who are familiar with the notion of double vector bundle in the sense of Pradines \cite{Pradines68} (or equivalently in the sense of Grabowski and Rotkiewicz \cite{Grabowski2009}) and closely-related notion of double Lie algebroid in the sense of K. Mackenzie \cite{Mackenzie1992} (or equivalently in the sense of T. Voronov \cite{voronov2012}), 
it is worth to remark that, just as Tulczyjew's original approach to particle dynamics, our extension to dynamics of charged particle rests on the following mathematical fact \cite{Dufour1991}: {\em a real vector bundle $E\to X$ and its dual vector bundle $E^*\to X$ are not canonically isomorphic, but their associated double vector bundles 
$(T^*E^*; E^*, E^{**}; X)$ and $(T^*E; E^*, E; X)$  are}.  

All ingredients involved in the present work have already appeared in the literature, but they will be reviewed here for completeness. It should be pointed out that, the main result obtained here would appear earlier (especially in Ref. \cite{Grabowska2006}) if enough attention was paid to Sternberg phase spaces. After communicating with J. Grabowski and P. Urba{\'n}ski, the author learned that there had been an approach to electrically charged particles based on affine geometries \cite{tulczyjew1992affine, urbanski2003affin, grabowska2004av}; this is associated with $\mathbb R$-principal bundles, but a reasonable extension to non-abelian gauge fields may also be possible.

\vskip 10pt
In Section \ref{Sternberg phase space} a detailed review of Sternberg phase space \cite{Sternberg77} is presented. For our purpose,  the presentation given here focuses more on the explicit local computations. In Section \ref{canonical isomorphism} a detailed review of the canonical isomorphism $T^*E^*\cong T^*E$ for any real vector bundle $E\to X$ is presented. This isomorphism is more natural when it is viewed as the canonical double vector bundle isomorphism from $(T^*E^*; E^*, E^{**}; X)$ onto $(T^*E; E^{*}, E; X)$; moreover, it enables Tulczyjew to bring the Lagrangian description of classical mechanics to the domain of symplectic geometry and unify Lagrangian and Hamiltonian formalisms of classical mechanics under a single geometric roof. In Section \ref{Special symplectic manifold} we review Tulczyjew's notion of {\em special symplectic manifold}, a concept more refined than that of symplectic manifold. In Section \ref{Tulczyjew triple} we first review the classical Tulczyjew triple used in the Tulczyjew's approach to particle dynamics and then introduce its magnetized version where Sternberg phase spaces play the role of $T^*X$. Section \ref{Tulczyjew's approach} contains new results. Here, the Hamiltonian side of Tulczyjew's approach, taken from Ref. \cite{Tulczyjew1974}, works for arbitrary symplectic manifold $M$, not just for the Sternberg phase space.  The magnetized version of Tulczyjew triple introduced in Section \ref{Tulczyjew triple} enables us to work out the Lagrangian side of Tulczyjew's approach to particles in gauge fields. In particular, for particles in Yang-Mills fields, a Lagrangian approach in the usual textbook way is derived, for which the Lagrangian is the sum of the ordinary Lagrangian $L(q, \dot q)$, the Lorentz term, and an {\em  extra new term}. A generalized version of Dirac's {\em charge quantization condition} \cite{Dirac1931} is also derived from this Lagrangian approach. (Please consult the Main Theorem on page \pageref{main theorem} for precise statement.) As far as this author knows, the aforementioned Lagrangian and charge quantization condition are new; cf. Ref. \cite{Balachandran1977} and example \ref{Wong's equations} on page \pageref{Wong's equations}.

For the convenience of readers, in the appendix we list some symbols used in this article, along with technically useful facts on tangent lift operator and on fiber bundle equipped with a connection. Most of symbols used in this article are quite standard in the mathematical literature, for example, $\tau_X$: $TX\to X$ denotes the tangent bundle of $X$, $\pi_X$: $T^*X\to X$ denotes the cotangent bundle of $X$, $\vartheta_X$ denotes the Liouville form on $T^*X$ and $\omega_X:=\mathrm{d}\vartheta_X$ denotes the tautological symplectic form on $T^*X$. Also, {\em for notational sanity in this article, we shall use the same notation for both a differential form (or a map) and its pullback under a fiber bundle projection map}. For example, for a differential form $\Omega$ on manifold $F$, its pullback under projection $X\times F\to F$ is also denoted by $\Omega$. 

\vskip 10pt
{\bf Acknowledgements.} The author learned Tulczyjew's elegant approach from Janusz Grabowski at the recent workshop on Geometry of Mechanics and Control Theory (India Institute of Sciences, Bangalore, India, January 2 - 10, 2014). Besides thanking Janusz Grabowski for his beautiful talk, he would also like to thank Partha Guha for organizing the wonderful workshop. Finally, he would like to thank John Baez, Janusz Grabowski, Jim Stasheff and Pawel Urba{\'n}ski for providing either valuable comments or additional references. 

\section{Sternberg phase space}\label{Sternberg phase space}
Throughout this section we assume that $X$ is a manifold, $G$ is a compact connected Lie group with Lie algebra $\mathfrak g$, $P\to X$ is a principal $G$-bundle with a fixed principal connection form $\Theta$, and $F$ is a hamiltonian $G$-space with symplectic form $\Omega$ and (equivariant) moment map $\Phi$. Recall that we shall use the same notation for both a differential form (or a map) and its pullback under a fiber bundle projection map. For example, the pullback of the moment map $\Phi$ under projection $X\times F\to F$ is also denoted by $\Phi$. 

Let $\mathcal F:=P\times_G F$ and $\mathcal F^\sharp$ be the pullback of diagram
\begin{eqnarray}\label{D: pullback}
\begin{tikzcd}[column sep=small]
 & &  \mathcal F\arrow{d}\\
T^*X\arrow{rr}{\pi_X}& &X
\end{tikzcd}\nonumber
\end{eqnarray}
Sternberg observed that \cite{Sternberg77}, with the above data, there is a correct substitute $\Omega_\Theta$ on $\mathcal F$ for $\Omega$ on $X\times F$, in the sense that $\Omega_\Theta$ is a closed real differential two-form on $\mathcal F$ and it is equal to $\Omega$ when
$P\to X$ is a trivial bundle with the product connection. He further observed that, if $\omega_X$ denote the canonical symplectic form on $T^*X$, then
$$
\omega_X+\Omega_\Theta
$$ is a symplectic form on $\mathcal F^\sharp$ --- the {\bf Sternberg symplectic form}. 

\vskip 10pt
To describe the Sternberg symplectic form, we need to do some preparations. For $a\in G$, the right action of $a$ on $P$ is denoted by $R_a$ and the adjoint action of $a$ on $\mathfrak g$ is denoted by ${\mathrm {Ad}}_a$. For $\xi\in\mathfrak g$, the infinitesimal right action of $\xi$ on $P$ is a vector field on $P$ and shall be denoted by $X_\xi$. Since $G$ is a compact connected Lie group, we can assume that it is a Lie subgroup of $\mathrm{SO}(N)$ for some positive integer $N$. Let us denote by $g$ the inclusion map of $G$ into the vector space of all real square matrices of order $N$. Note that, in terms of $g$, the Maurer-Cartan form on $G$ can be written as $g^{-1}\, \mathrm{d}g$, here, the product between $g^{-1}$ and $\mathrm{d}g$ is the matrix multiplication.

\subsection{Principal connection} Let $\Theta$ be a $\mathfrak g$-valued differential one-form on $P$. Then $\Theta$ is/defines a principal connection on the principal $G$-bundle $P\to X$ if it satisfies the following two conditions\footnote{Our $\Theta$ here is the negative of the $\Theta$ in Ref. \cite{Sternberg77} because $R_a$ is right multiplication by $a$ here and by $a^{-1}$ in Ref. \cite{Sternberg77}.}: 
\begin{center}
1) ${R_{a^{-1}}}^*\,\Theta ={\mathrm {Ad}}_a\Theta$ for any $a\in G$, \quad\quad 2) $\Theta(X_\xi)=\xi$ for any $\xi\in \mathfrak g$.
\end{center}
Working locally, we may assume that $P\to X$ is trivial. Suppose that 
\begin{eqnarray}
\begin{tikzcd}[column sep=small]
X\times G \arrow{rr}{\phi}[swap]{\cong} \arrow{dr}& &  P\arrow{dl}\\
& M & 
\end{tikzcd} \nonumber
\end{eqnarray} is a trivialization, then, as a $\mathfrak g$-valued differential one-form on $X\times G$, 
\begin{eqnarray}
\fbox{$\phi^*\Theta=g^{-1}A_\phi g+g^{-1}\, \mathrm{d}g$}
\end{eqnarray}
for a unique $\mathfrak g$-valued differential one-form $A_\phi$ on $X$. Similarly, if $\phi'$ is another trivialization, we have
$$
{\phi'}^*\Theta=g^{-1}A_{\phi'}g+g^{-1}\, {\mathrm d}g
$$ for a unique $\mathfrak g$-valued differential one-form $A_{\phi'}$ on $X$. To see how $A_\phi$ and $A_{\phi'}$ are related, we note that the bundle isomorphism $\lambda$ defined by the commutative triangle
\begin{eqnarray}
\begin{tikzcd}[column sep=small]
& P  &\\
X\times G  \arrow{ur}{\phi}  \arrow{rr}{\lambda}  & &  X\times G  \arrow{ul}[swap]{\phi'} 
\end{tikzcd}\nonumber
\end{eqnarray} can be written as  $\lambda(x, b)=(x, a(x)b)$ for a unique smooth map $a$: $X\to G$. Since $\phi^*\Theta=\lambda^*{\phi'}^*\Theta$, we have
\begin{eqnarray}\label{gaugeT}
A_\phi=a^{-1}A_{\phi'}a+a^{-1}\, {\mathrm d}a\quad\mbox{or}\quad \fbox{$A_{\phi'}=aA_\phi a^{-1}+a\, {\mathrm d}a^{-1}$}.
\end{eqnarray}

\subsection{Sternberg form $\Omega_\Theta$ on $\mathcal F$}
For the hamiltonian $G$-space $F$, recall that $\Omega$ is its symplectic form and $\Phi$: $F\to\mathfrak g^*$ is its moment map. We let $Y_\xi$ be the vector field on $F$ which represents the infinitesimal left action of $\xi\in \mathfrak g$ on $F$. Then $\Omega$ is invariant under the $G$-action on $F$, $\Phi$ is $G$-equivariant, and  
\begin{eqnarray}\label{phi1}
Y_\xi  \lrcorner\,\Omega=\langle\xi, {\mathrm d}\Phi\rangle\quad\mbox{for any $\xi\in\mathfrak g$}.
\end{eqnarray}
Here $\lrcorner$ denotes the interior product and  $\langle\, ,\, \rangle$ denotes the paring of elements in $\mathfrak g$ with elements in $\mathfrak g^*$. In view of the fact that $\Phi$ is $G$-equivariant, Eq. \eqref{phi1} implies that
\begin{eqnarray}\label{phi2}
\Omega(Y_{\xi_1}, Y_{\xi_2})=\langle [\xi_1, \xi_2], \Phi\rangle \quad\mbox{for any $\xi_1, \xi_2\in\mathfrak g$}.
\end{eqnarray} 

Denote by $\phi_F$ the composition map $X\times F\cong (X\times G)\times_G F\buildrel \phi\times_G F\over \longrightarrow  P\times_G F=:\mathcal F$. If we let $\lambda_F$ be the fiber bundle isomorphism in the commutative triangle 
\begin{eqnarray}
\begin{tikzcd}[column sep=small]
& \mathcal F &\\
X\times F \arrow{rr}{\lambda_F}  \arrow{ur}{\phi_F} & &  X\times F  \arrow{ul}[swap]{\phi'_F} , 
\end{tikzcd}\nonumber
\end{eqnarray}
then $\lambda_F(x, f)=(x, a(x)\cdot f)$. Let 
\[\fbox{$\Omega_\phi:=\Omega-{\mathrm d}\langle A_\phi, \Phi\rangle$}.\] 
The following lemma implies that there is a well-defined closed real differential two-form $\Omega_\Theta$ on $ \mathcal F$, referred to as the {\bf Sternberg form} on $\mathcal F$, such that $\Omega_\phi=\phi_F^*\Omega_\Theta$. 
\begin{lemma}
With the notations as above, we have $\lambda_F^*\, \Omega_{\phi'}=\Omega_\phi$. Consequently, as closed real differential two-forms on $\mathcal F$, $({\phi_F}^{-1})^*\Omega_\phi=({\phi'_F}^{-1})^*\Omega_{\phi'}$. 
\end{lemma}
Since $\Omega_{\phi'}=\Omega-{\mathrm d}\langle A_{\phi'}, \Phi\rangle$ and $\Omega_{\phi}=\Omega-{\mathrm d}\langle A_{\phi}, \Phi\rangle$, we have
\begin{eqnarray}
\lambda_F^*\Omega_{\phi'}&=&\lambda_F^*\Omega-{\mathrm d}\langle A_{\phi'}, L_a^*\Phi\rangle\quad \mbox{here $L_a$ is the left action of $a$ on $F$}\cr
&=& \lambda_F^*\Omega-{\mathrm d}\langle A_{\phi'}, {{\mathrm {Ad}}_{a^{-1}}}^*\Phi\rangle\quad \mbox{because $\Phi$ is $G$-equivariant}\cr
&=& \lambda_F^*\Omega-{\mathrm d}\langle {\mathrm {Ad}}_{a^{-1}} A_{\phi'},\Phi\rangle\cr
&=& \lambda_F^*\Omega-{\mathrm d}\langle A_{\phi},\Phi\rangle + {\mathrm d}\langle a^{-1}{\mathrm d}a,\Phi\rangle \quad \mbox{using Eq. \eqref{gaugeT}}\cr
&=&\Omega_{\phi}+ \lambda_F^*\Omega-\Omega+{\mathrm d}\langle a^{-1}{\mathrm d}a,\Phi\rangle, \nonumber
\end{eqnarray}
so the above lemma is equivalent to
\begin{claim} Let $\lambda_F$: $X\times F\to X\times F$, $\Phi$: $F\to \mathfrak g^*$, $a$: $X\to G$, and $\Omega$ be as before. Then $\lambda_F^*\Omega=\Omega - {\mathrm d}\langle a^{-1}{\mathrm d}a,\Phi\rangle$.
\end{claim}
\begin{proof} Let $\xi=\mathrm{d} a\, a^{-1}$. Since $\lambda_F(x, f)=(x, a(x)\cdot f)$, we have
\begin{eqnarray}
T_{(x, f)}\lambda_F: \;T_xX \times T_f F &\to& T_xX\times T_{a(x)\cdot f} F\cr
(u, v) &\mapsto & (u, T_fL_{a(x)}(v)+(Y_{u\lrcorner \xi|_x})|_{a(x)\cdot f}).\nonumber
\end{eqnarray}
With the understanding that $\Omega$ represents both the symplectic form on $F$ and its pullback under projection $X\times F\to F$, for $(u_1, v_1)$, $(u_2, v_2)$ in $T_{(x, f)}(X\times F)$, $$(\lambda_F^*\Omega)|_{(x,f)}((u_1, v_1), (u_2, v_2))$$ is equal to 
\begin{eqnarray}
 &&(L_{a(x)}^*\Omega)|_f (v_1, v_2)\cr
 &&+\, \Omega (Y_{u_1\lrcorner\,\xi|_x}, Y_{u_2\lrcorner\,\xi|_x})|_{a(x)\cdot f}\cr
 &&+\, \Omega|_{a(x)\cdot f}(Y_{u_1\lrcorner\,\xi|_x}|_{a(x)\cdot f}, T_fL_{a(x)}(v_2))-\,\Omega|_{a(x)\cdot f}(Y_{u_2\lrcorner\,\xi|_x}|_{a(x)\cdot f}, T_fL_{a(x)}(v_1)).\nonumber
\end{eqnarray}
Let $\eta=a^{-1}\, {\mathrm d}a$. Then, in view of the fact that $\Phi$ is $G$-equivariant and $\eta=\mathrm{ad}_{a^{-1}}\xi$, the above expression becomes
\begin{eqnarray}
 && \Omega|_f (v_1, v_2)\quad\mbox{because $\Omega$ is $G$-invariant}\cr
 &&+\langle [u_1\lrcorner\,\eta|_x, u_2\lrcorner\,\eta|_x], \Phi|_f\rangle  \quad \mbox{using Eq. \eqref{phi2}}\cr
 &&+\langle u_1\lrcorner\,\eta|_x, v_2\lrcorner\,{\mathrm d}\Phi|_f\rangle - \langle u_2\lrcorner\,\eta|_x, v_1\lrcorner\,{\mathrm d}\Phi|_f\rangle \quad \mbox{using Eq. \eqref{phi1}}\cr
 &=& \Omega|_f (v_1, v_2)+\langle \eta^2|_x, \Phi|_f\rangle(u_1, u_2)+\langle \eta|_x, {\mathrm d}\Phi|_f\rangle((u_1, v_1), (u_2, v_2))\cr
 &=& (\Omega-{\mathrm d}\langle \eta, \Phi\rangle)|_{(x, f)}((u_1, v_1), (u_2, v_2))\quad\mbox{because $d\eta =-\eta^2$}.\nonumber
 \end{eqnarray}
 Therefore $\lambda_F^*\Omega=\Omega - {\mathrm d}\langle \eta,\Phi\rangle = \Omega - {\mathrm d}\langle a^{-1}{\mathrm d}a,\Phi\rangle$.
\end{proof}

\subsection{Sternberg phase space}
With the Sternberg form $\Omega_\Theta$ on $\mathcal F$ and the canonical symplectic form $\omega_X$ on $T^*X$ being both closed, and the fact from the definition of $\mathcal F^\sharp$ that there are fiber bundle projections from $\mathcal F^\sharp$ to $\mathcal F$ and also to $T^*X$, we know that
\begin{eqnarray}
\fbox{$\omega_\Theta:=\omega_X+\Omega_\Theta$}
\end{eqnarray} is a closed real differential two-form on $\mathcal F^\sharp$.
\begin{claim}
With the notations as above, $\omega_\Theta$ is non-degenerate everywhere on $\mathcal F^\sharp$, so it is a symplectic form on $\mathcal F^\sharp$. 
\end{claim}
\begin{proof}Introducing local coordinate functions $(q^i, p_j)$ on $T^*X$ and $z^\alpha$ on $F$, and denoting $\partial\over \partial q^i$ by $\partial_i$,   $\partial\over \partial z^\alpha$ by $\partial_\alpha$, then $\omega_\Theta$ can be locally represented by
\begin{eqnarray}
{\mathrm d}p_i\wedge {\mathrm d}q^i+{1\over 2}\Omega_{\alpha\beta}\, \mathrm{d}z^\alpha\wedge \mathrm{d}z^\beta-{1\over 2}\langle \partial_iA_j-\partial_jA_i, \Phi\rangle\, {\mathrm d}q^i\wedge {\mathrm d}q^j + \langle A_i, \partial_{\alpha}\Phi\rangle\, {\mathrm d}q^i\wedge \mathrm{d}z^\alpha\nonumber
\end{eqnarray}
which is then easy to see to be non-degenerate everywhere. Therefore, $\omega_\Theta$ is a symplectic form on $P^\sharp$.
\end{proof}
In summary, we have
\begin{theorem}[Sternberg, 1977]
With the data and notations given in the beginning of this section, we have the following statements.

1) There is a closed real differential two-form $\Omega_\Theta$ on $\mathcal F$ which is of the form $\Omega-\mathrm{d}\langle A, \Phi\rangle$ under a local trivialization of $P\to X$ in which the connection form $\Theta$ is represented by the $\frak g$-valued differential one-form $A$ on $X$.

2) The differential two-form $\omega_\Theta:=\omega_X+\Omega_\Theta$ is a symplectic form on $\mathcal F^\sharp$.
\end{theorem} 
The symplectic manifold $(\mathcal F^\sharp, \omega_\Theta)$ is referred to as a {\bf Sternberg phase space}. In particular $(T^*X, \omega_X)$ is a Sternberg phase space. In Ref. \cite{Weinstein78} A. Weinstein introduced a symplectic space out of the principal $G$-bundle $P\to X$ and the hamiltonian $G$-space $F$, and showed that a connection $\Theta$ on $P\to X$ yields a symplectomorphism from his symplectic space to the Sternberg phase space $(\mathcal F^\sharp, \omega_\Theta)$, a reason for A. Weinstein to call his symplectic space the {\em universal phase space}. 

\section{A canonical isomorphism of double vector bundles}\label{canonical isomorphism}
The purpose of this section is to give a detailed review of the following simple and elegant mathematical fact \cite{Dufour1991}: {\em a real vector bundle $E\to X$ and its dual vector bundle $E^*\to X$ are not canonically isomorphic, but $T^*E^*$ and $T^*E$ are canonically isomorphic as symplectic manifolds}. 

Let $\pi$: $E\to X$ be a real vector bundle and $\pi^*$: $E^*\to X$ be its dual vector bundle. Consider the diagram
\begin{eqnarray}\label{dia}
\begin{tikzcd}
T^*X \arrow{dr}{\pi_X} &E^* \arrow{d}{\pi^*}\\
E\arrow{r}{\pi}& X
\end{tikzcd}
\end{eqnarray}
The limit of diagram \eqref{dia} exists and is unique up to diffeomorphisms, in fact, it is the total space of the Whitney sum of the three vector bundles over $X$. We shall show that both $T^*E^*$ and $T^*E$ can be this limit, so they must be diffeomorphic to each other. The detailed arguments are given below.

\underline{Step one}. For any $e\in E$, we have injective linear map $E_{\pi(e)}\cong T_eE_{\pi(e)}\subset T_eE$ whose dual is a surjective linear map $T^*_eE\to E^*_{\pi(e)}$ which, upon being globalized, becomes the top arrow $p_E$ in the commutative square 
\begin{eqnarray}\label{doublevb}
\begin{tikzcd}
T^*E \arrow{r}{p_E}\arrow{d}[swap]{\pi_E} &E^* \arrow{d}{\pi^*}\\
E\arrow{r}{\pi}& X.
\end{tikzcd} 
\end{eqnarray}

\underline{Step two}.  Choosing a connection on $E\to X$, then we have the commutative square
\begin{eqnarray}
\begin{tikzcd}
T^*E \arrow{r}{T^*_\pi}\arrow{d}[swap]{\pi_E} &T^*X \arrow{d}{\pi_X}\\
E\arrow{r}{\pi}& X\nonumber
\end{tikzcd} 
\end{eqnarray} by appendix \ref{A: fact}.

\underline{Step three}. Combining steps one and two, we have commutative diagram
\begin{eqnarray}\label{limit}
\begin{tikzcd}
T^*E \arrow[bend left]{drr}{p_E}
\arrow[bend right]{ddr}[swap]{\pi_E}
\arrow{dr}{T^*_\pi} & & \\
& T^*X \arrow{rd}{\pi_X}& E^* \arrow{d}{\pi^*} \\
& E \arrow{r}{\pi} & X
\end{tikzcd}
\end{eqnarray}
which in turn yields a smooth map $T^*E\to T^*X\oplus E\oplus E^*$, fibering over $X$. This smooth map is a bijection because  $T^*_eE\to T^*_{\pi(e)}X \times \{e\}\times E^*_{\pi (e)}$ is a bijection for each $e\in E$. Moreover, by using the local triviality of $E\buildrel\pi\over\to X$, one can check that this map is a diffeomorphism, so it turns $T^*E\to X$ into a vector bundle over $X$. This vector bundle is isomorphic to $T^*X\oplus E\oplus E^{*}$, but it is not canonical because of its dependence on the choice of a connection on $E\buildrel\pi\over\to X$.  

\underline{Step four}.  Since $E^{**}\cong E$ naturally and a connection on $E\buildrel\pi\over\to X$ is turned into a connection on $E^*\buildrel\pi^*\over\to X$ upon taking dual, replacing $E$ by $E^*$ in the above analysis, commutative diagram \eqref{limit} becomes commutative diagram
\[
\begin{tikzcd}
T^*E^* \arrow[bend left]{drr}{p_{E^*}}
\arrow[bend right]{ddr}[swap]{\pi_{E^*}}
\arrow{dr}{T^*_{\pi^*}} & & \\
& T^*X \arrow{rd}{\pi_X}& \hskip 25pt E^{**}\cong E \arrow{d}{\pi^{**}} \\
& E^* \arrow{r}{\pi^*} & X
\end{tikzcd}\] So we have vector bundle isomorphism $T^*E^*\cong T^*X\oplus E^*\oplus E^{**}$ by the same argument as above. 

\underline{Step five}. Define $\kappa$ via commutative diagram
\begin{eqnarray}\label{Definition: kappa}
\begin{tikzcd}
T^*E^* \arrow{rr}{T^*_{\pi^*}\oplus \pi_{E^*}\oplus p_{E^*}}[swap]{\cong}\arrow{d}[swap]{\kappa}{\cong} &&T^*X\oplus E^*\oplus E^{**}\arrow{d}{1\oplus 1\oplus - \iota^{-1}}[swap]{\cong}\\
T^*E \arrow{rr}{T^*_{\pi}\oplus p_E\oplus \pi_E}[swap]{\cong} &&T^*X\oplus E^*\oplus E
\end{tikzcd} 
\end{eqnarray} where $-\iota^{-1}$ is the negative of the inverse of the natural vector bundle identification
\begin{eqnarray}\label{identification}
\iota: E&\longrightarrow & E^{**}\cr
u &\mapsto & \alpha\mapsto\langle u, \alpha\rangle.
\end{eqnarray} We shall see later that,  $\kappa$ is a symplectomorphism and is independent of the choice of a connection on $E\to X$.

\vskip 10pt
In summary, once a connection on $E\buildrel\pi\over\to X$ is chosen, $T^*E$ and $T^*E^*$ become vector bundles over $X$; moreover, we have vector bundle isomorphisms $T^*E\cong T^*X\oplus E\oplus E^*$ and $T^*E^*\cong T^*X\oplus E^*\oplus E^{**}$. Since $E\oplus E^*\cong E^*\oplus E^{**}$, we have vector bundle isomorphism $\kappa$: $T^*E^*\to T^*E$ over $X$. 
While the vector bundle structures on $T^*E\to X$ and $T^*E^*\to X$ depend on the choice of a connection on $E\buildrel\pi\over\to X$, the diffeomorphism $\kappa$ does not if we identify $E\oplus E^*$ with $E^*\oplus E^{**}$ via
$$
 \begin{pmatrix}
 0&1\cr
- \iota &0
 \end{pmatrix}
$$ 
where $\iota$ is the map defined in Eq. \eqref{identification}. Since $\pi_{E^*}$: $T^*E^*\to E^*$ is a vector bundle and the definition of $\kappa$ in diagram \eqref{Definition: kappa} makes triangle 
\begin{eqnarray}
\begin{tikzcd}[column sep=small]
T^*E^* \arrow{rr}{\kappa}[swap]{\cong} \arrow{dr}[swap]{\pi_{E^*}} & &  T^*E \arrow{dl}{p_E}\\
& E^* & \nonumber
\end{tikzcd}
\end{eqnarray} 
commutative, the canonical diffeomorphism $\kappa$ turns $p_E$: $T^*E\to E^*$ into a canonical vector bundle so that $\kappa$ becomes a canonical vector bundle isomorphism over $E^*$ and $(T^*E; E, E^*; X)$ as in diagram \eqref{doublevb} becomes a (canonical) {\em double vector bundle} (with $T^*X\buildrel\pi_X\over\to X$ as its core) in the sense of J. Pradines \cite{Pradines68}. 

In short, {\em to any real vector bundle $E\buildrel\pi\over\to X$, there associate two canonically isomorphic double vector bundles $(T^*E; E, E^*; X)$ and $(T^*E^*; E^*, E^{**}; X)$:
\[
\begin{tikzcd}[row sep=scriptsize, column sep=scriptsize]
&T^*E^* \arrow{dl}[swap]{\pi_{E^*}} \arrow{rr}{p_{E^*}} \arrow[dashed]{dd}& & E^{**} \arrow{dl}{\pi^{**}} \arrow{dd}[description]{-\iota^{-1}}\\ 
E^* \arrow[crossing over]{rr}{\hskip 15pt\pi^*} \arrow{dd}[description]{1} & & X \\
& T^*E \arrow{dl}[swap] {p_E} \arrow{rr}{\hskip -25pt \pi_E} & & E \arrow{dl}{\pi} \\
E^* \arrow{rr}{\pi^*} & & X \arrow[crossing over, leftarrow]{uu}[description]{1}
\end{tikzcd}
\] where the dashed arrow is $\kappa$. Moreover $\kappa$ is a symplectomorphism.} This isomorphism of double vector bundles shall be referred to as {\bf the canonical isomorphism}.

\vskip 10pt
It remains to show that $\kappa$ is a canonical symlectomorphism, i.e., it is a symlectomorphism and is independent of the choice of a connection on $E\buildrel\pi\over\to X$. That will be clear after we work out a local representation for $\kappa$.

\subsection{Local formulae}
Let $n=\dim X$ and $k$ be the rank of $E\buildrel\pi\over\to X$. Since we work locally, we may assume that $X$ is diffeomorphic to $\mathbb R^n$ and $E\buildrel\pi\over\to X$ is trivial. Let us fix a diffeomorphism $Q$: $X\to \mathbb R^n$ and a trivialization 
\begin{eqnarray}
\begin{tikzcd}[column sep=small]
E \arrow{rr}{\phi}[swap]{\cong} \arrow{dr}[swap]{\pi} & &  X\times \mathbb R^k \arrow{dl}{p_1}\\
& X & \nonumber
\end{tikzcd}
\end{eqnarray} where $p_1$ is the projection onto the first factor. Then we have diffeomorphisms $$E\cong X\times \mathbb R^k\cong \mathbb R^n\times \mathbb R^k, \quad E^*\cong X\times (\mathbb R^k)^*\cong \mathbb R^n\times \mathbb R^k$$ 
which shall be denoted by $(q, u)$ and $(q, \alpha)$ respectively, and diffeomorphisms 
\begin{eqnarray}
T^*E &\cong& \mathbb R^n\times \mathbb R^k\times (\mathbb R^n)^*\times (\mathbb R^k)^*\cong (\mathbb R^n\times \mathbb R^k)^2, \cr 
T^*E^*&\cong& \mathbb R^n\times (\mathbb R^k)^*\times (\mathbb R^n)^*\times (\mathbb R^k)^{**} \cong (\mathbb R^n\times \mathbb R^k)^2\nonumber
\end{eqnarray}
which shall be denoted by $(q, u, p, \alpha)$ and $(q, \alpha, p, \hat u)$ respectively. 

Under the trivialization $\phi$, the connection on $E\buildrel\pi\over\to X$ is represented by a real $k\times k$-matrix valued differential one-form $A\cdot \mathrm{d}q :=A_i\,\mathrm{d}q^i $ on $X$, according to formula \eqref{localT^*_f}, the map $T^*E\buildrel T^*_\pi\over\to T^*X$ can be represented by
\begin{eqnarray}
(q, u, p, \alpha)\mapsto (q, p- \alpha\cdot Au)
\end{eqnarray} where $\cdot$ means the dot product of $\mathbb R^k$, and $A$ is viewed as a local function on $T^*Y$. On the other hand, the connection on $E^*\buildrel\pi^*\over\to X$ is represented by $-A^T$ on $X$, so, according to formula \eqref{localT^*_f}, the map $T^*E^*\buildrel T^*_{\pi^*}\over\to T^*X$ can be represented by
\begin{eqnarray}
(q,  \alpha, p, \hat u)\mapsto (q, p+\hat u\cdot A^T\alpha)=(q, p+\alpha\cdot A\hat u)).
\end{eqnarray}
Note also that, the map $T^*E\buildrel \pi_E\over \to E$ and $T^*E\buildrel p_E\over \to E^*$ are represented by
$$
(q, u, p,  \alpha)\mapsto (q, u), \quad (q, u, p, \alpha)\mapsto (q, \alpha)
$$ respectively, the map $T^*E^*\buildrel \pi_{E^*}\over\to E^*$ and $T^*E^*\buildrel p_{E^*}\over\to E^{**}$ are represented by
$$
(q, \alpha, p, \hat u)\mapsto (q, \alpha), \quad (q, \alpha, p, \hat u)\mapsto (q, \hat u)
$$ respectively, and $\iota$: $E\to E^{**}$ is represented by $(q, u)\mapsto (q, u)$.  Therefore, from the definition of $\kappa$ in diagram \eqref{Definition: kappa}, we conclude that Tulczyjew isomorphism $\kappa$: $T^*E^*\to T^*E$ is represented by
\begin{eqnarray}
(q,  \alpha, p, \hat u)\mapsto (q,  -\hat u, p, \alpha).
\end{eqnarray}

This local representation of $\kappa$ immediately implies that $\kappa$ preserves the natural symplectic structures, is independent of the choice of connection on $E\buildrel\pi\over\to X$, and fibers over both $E^*$ and $E$, i.e., both triangle  
\begin{eqnarray}
\begin{tikzcd}[column sep=small]
T^*E^* \arrow{rr}{\kappa}[swap]{\cong} \arrow{dr}[swap]{\pi_{E^*}} & &  T^*E \arrow{dl}{p_E}\\
& E^* & \nonumber
\end{tikzcd}
\end{eqnarray} and square 
\begin{eqnarray}
\begin{tikzcd}[column sep=small]
T^*E^* \arrow{rr}{\kappa}[swap]{\cong} \arrow{d}[swap]{p_{E^*}} & &  T^*E \arrow{d}{\pi_E}\\
E^{**}\arrow{rr}{-\iota^{-1}}[swap]{\cong} & & E
\end{tikzcd}\nonumber
\end{eqnarray} are commutative, also an easy fact from commutative diagram \eqref{Definition: kappa}. 

\section{Special symplectic manifold}\label{Special symplectic manifold}
In mid-1970s W. M. Tulczyjew discovered an approach to classical mechanics which brings the Hamiltonian formalism and the Lagrangian formalism under a common geometric roof: the dynamics of a particle with configuration space $X$ is determined by a Lagrangian submanifold $D$ of $TT^*X$ (the total tangent space of $T^*X$), and the description of $D$ by its Hamiltonian $H$: $T^*X\to \mathbb R$ (resp. its Lagrangian $L$: $TX\to\mathbb R$) yields the Hamilton (resp. Euler-Lagrange) equation.

To formulate this approach to mechanics, Tulczyjew introduced the notion of {\em special symplectic manifold}. He observed that, on $TT^*X$, there is one symplectic manifold structure and two special symplectic manifold structures (refereed to as {\em Liouville structures} in Ref. \cite{tulczyjew2008liouville, tulczyjew1999slow}); therefore, for a classical particle with configuration space $X$ under a given conservative force,  there is one dynamics (i.e., the submanifold $D$ which is Lagrangian with respect to the symplecic structure on $TT^*X$) and two descriptions of this dynamics (i.e., the description of $D$ via the two Liouville structures on $TT^*X$).  

\vskip 10pt 
Let $(P, \omega)$ be a symplectic manifold with symplectic form $\omega$, $N$ be a submanifold of $P$. We say that $N$ is an {\em isotropic submanifold} of $(P, \omega)$ if 
the pullback of $\omega$ under the inclusion $N\hookrightarrow P$ is identically zero, and is an
{\em Lagrangian submanifold} of $(P, \omega)$ if it is isotropic and $\dim P=2\dim N$.
\begin{definition}\label{D: SpecialSM}
A {\em special symplectic manifold} is a quadruple $(P, M, \pi, \vartheta)$, where $(P, M, \pi)$ is a smooth fiber bundle, $\vartheta$ is a differential one-form on $P$, and there is a diffeomorphism $\alpha$: $P\to T^*M$ such that $\pi=\pi_M\circ \alpha$ and $\vartheta=\alpha^*\vartheta_M$. 
\end{definition}
The diffeomorphism $\alpha$ is unique, assuming it exists. If $(P, M, \pi, \vartheta)$ is a special symplectic manifold,  then $(P, \mathrm{d}\vartheta)$ is a symplectic manifold isomorphic to $(T^*M, \omega_M)$, and is called the {\bf underlying symplectic manifold} of $(P, M, \pi, \vartheta)$. 

The following proposition is obvious.
\begin{proposition}\label{A: HamitonianSSS}
1) Let $(M, \omega)$ be a symplectic manifold. Then $(TM, M, \tau_M, \mathrm{i}_T\omega)$ is a special symplectic manifold.

2) Let $(P_1, M_1, \pi_1, \vartheta_1)$ and $(P_2, M_2, \pi_2, \vartheta_2)$ be special symplectic manifolds. Then the quadruple $(P_2\times P_1,M_2\times M_1, \pi_2\times \pi_1, \vartheta_2-\vartheta_1)$ is a special symplectic manifold.

\end{proposition}
The special symplectic manifold in part 1) is referred to as a {\bf Hamiltonian special symplectic manifold}. For us, the interesting $M$ is $T^*X$ or more generally a Sternberg phase space. 

The following two propositions are taken from Ref. \cite[Section 3]{Tulczyjew1972}. 
\begin{proposition}
Let $(P, M, \pi, \vartheta)$ be a special symplectic manifold, $K$ a submanifold of $M$ and $F$ a smooth real function on $K$. Then the set
$$
N:=\{p\in\pi^{-1}(K)\mid \iota^*_{\pi^{-1}(K)}\vartheta = (\pi|_{\pi^{-1}(K})^*\mathrm{d}F\; \mbox{at $p$}\}$$ is a Lagrangian submanifold of $(P, \mathrm{d}\vartheta)$, $K=\pi(N)$, the mapping $\varrho$ defined by the commutative diagram 
\begin{eqnarray}
\begin{tikzcd}[column sep=small]
N \arrow[hook]{rr}{\iota_N}\arrow{d}[swap]{\varrho} & &  P \arrow{d}{\pi}\\
K \arrow[hook]{rr}{\iota_K} & & M\nonumber
\end{tikzcd}
\end{eqnarray}
is a submersion, the fibers of $\varrho$ are connected and $\iota_N^*\vartheta = \varrho^*\mathrm{d}F$.
\end{proposition}
The Lagrangian submanifold $N$ in this proposition is called the Lagrangian submanifold {\em generated} by $F$, and $F$ is called a {\em generating function} of $N$.  
\begin{proposition}\label{P: gamma}
Let $(P, M, \pi, \vartheta)$ be a special symplectic manifold, $K$ a submanifold of $M$ and $N$ an isotropic submanifold of $(P, \mathrm{d}\vartheta)$ such that $K:=\pi(N)$ is a submanifold of $M$, the mapping $\varrho$ defined by the commutative diagram 
\begin{eqnarray}
\begin{tikzcd}[column sep=small]
N \arrow[hook]{rr}{\iota_N}\arrow{d}[swap]{\varrho} & &  P \arrow{d}{\pi}\\
K \arrow[hook]{rr}{\iota_K} & & M\nonumber
\end{tikzcd}
\end{eqnarray}
is a submersion and the fibers of $\varrho$ are connected. Then there is a unique closed differential form $\gamma$ on $K$ such that $\iota_N^*\vartheta = \varrho^*\gamma$. If $\gamma$ is exact and $\gamma=\mathrm{d}F$, then $N$ is  contained in the Lagrangian submanifold generated by $F$, and if $N$ is a Lagrangian submanifold generated by a function $F$, then $\mathrm{d}F =\gamma$.
\end{proposition}

\section{Tulczyjew triple and its magnetized version}\label{Tulczyjew triple}
We shall split the discussion of the Tulczyjew triple into two parts: the classical Tulczyjew triple and the magnetized Tulczyjew triple, though the former is a special case of the later. 

\subsection{Classical Tulczyjew triple}
Since $(T^*X, \omega_X)$ is a symplectic manifold, there is a vector bundle isomorphism $\beta_X$: $TT^*X\to T^*T^*X$ over $T^*X$. Let $\alpha_X=\kappa\circ \beta_X$, where $\kappa$: $T^*T^*X\to T^*TX$ is the canonical isomorphism reviewed in Section \ref{canonical isomorphism}, then we have {\em Tulczyjew triple}
\begin{eqnarray}
\begin{tikzcd}[column sep=small]
T^*T^*X && TT^*X \arrow{ll}[swap]{\beta_X}{\cong}  \arrow{rr}{\alpha_X}[swap]{\cong} & &  T^*TX .
\end{tikzcd}
\end{eqnarray} 
Let
\begin{eqnarray}
\vartheta_X^H:=\beta_X^*\vartheta_{T^*X}, \quad \vartheta_X^L:=\alpha_X^*\vartheta_{TX}.
\end{eqnarray}

Since $T^*X$ is a symplectic manifold, by Proposition \ref{A: HamitonianSSS}, we have a Hamiltonian special symplectic manifold $$\fbox{$S^H_X:=(TT^*X, T^*X, \tau_{T^*X}, \vartheta^H_X)$}$$ as usual. With the help of local formula \eqref{local: alpha0}, one can check that diagram
\begin{eqnarray}\label{D: commutativeSQ}
\begin{tikzcd}[column sep=small]
TT^*X \arrow{rr}{\alpha_X}\arrow{d}[swap]{T\pi_X} & &  T^*TX \arrow{d}{\pi_{TX}}\\
TX\arrow{rr}{1}& &TX
\end{tikzcd}
\end{eqnarray} is commutative, so quadruple $$\fbox{$S^L_X:=(TT^*X, TX, T\pi_X, \vartheta^L_X)$}$$ is a special symplectic manifold. We shall call $S_X^L$ a {\bf Lagrangian special symplectic manifold}.  

\subsubsection{Local formulae}
In Section \ref{canonical isomorphism}, if we take $E\to X$ to be $TX\to X$, then the canonical isomorphism $\kappa$: $T^*T^*X\to T^*TX$ has this local representation:
$$(q, p, \dot p, {\dot q})\mapsto (q, -{\dot q}, \dot p, p),$$ 
so we have
\begin{eqnarray}\label{local: alpha0}
 \fbox{$\beta_X(q, p, \dot q, \dot p) = (q, p, \dot p, -\dot q),\quad \alpha_X(q, p, \dot q, \dot p) = (q, \dot q, \dot p, p)$}
\end{eqnarray} in local representation.
Therefore, locally $\vartheta^H_X=\dot p_i\, {\mathrm d}q^i-\dot q^i\, {\mathrm d}p_i$, $\vartheta^L_X=\dot p_i\, {\mathrm d}q^i+p_i\, {\mathrm d}\dot q^i$. In more compact form, we have
\begin{eqnarray}
\fbox{$\vartheta^H_X=\dot p\cdot {\mathrm d}q-\dot q\cdot {\mathrm d}p, \quad \vartheta^L_X=\dot p\cdot {\mathrm d}q+p\cdot {\mathrm d}\dot q$}
\end{eqnarray} locally. In view of the fact that locally $\hat\vartheta_X=p\cdot\dot q$, we have
\begin{eqnarray}
\fbox{$\vartheta^L_X-\vartheta^H_X=\mathrm{d}\hat \vartheta_X$.}
\end{eqnarray}
Then \fbox{$\mathrm{d}\vartheta^L_X=\mathrm{d}\vartheta^H_X=:\Omega_X$}, so $\vartheta^L_X$ and $\vartheta^H_X$ yields the same symplectic form $\Omega_X$ on $TT^*X$. Note that locally
\begin{eqnarray}
\fbox{$\Omega_X=\mathrm{d}\dot p_i\wedge \mathrm{d}q^i+\mathrm{d}p_i\wedge \mathrm{d}\dot q^i$.}
\end{eqnarray}
In terms of operators $\mathrm{i}_T$ and $\mathrm{d}_T$ on page \pageref{tangent lift} for tangent bundle $TT^*X\to T^*X$, we have
\begin{eqnarray}
\fbox{$\hat \vartheta_X =\mathrm{i}_T\vartheta_X, \quad \vartheta_X^H =\mathrm{i}_T\omega_X, \quad \vartheta_X^L =\mathrm{d}_T\vartheta_X, \quad \Omega_X =\mathrm{d}_T\omega_X$}
\end{eqnarray}whose validity can be checked by local computations based on formulae on page \pageref{tangent lift}.

\vskip 10pt
In summary, there are two special symplectic structures on $TT^*X$ that underlie the symplectic structure $\Omega_X$ on $TT^*X$, the one that corresponds to $S^H_X$ is called the {\bf  Hamiltonian special symplectic structure}, and the one that corresponds to $S^L_X$ is called the {\bf Lagrangian special symplectic structure}. When we go from $T^*X$ to a generic symplectic manifold $M$, the Hamiltonian special symplectic structure still exists (on $TM$) as usual, but the Lagrangian special symplectic structure ceases to exist. A key observation of this article is that {\em the Lagrangian special symplectic structure still exist if we go from $T^*X$ to a Sternberg phase space.}

\subsection{Magnetized Tulczyjew triple}\label{MagnetizedTT}
In this subsection we shall assume that $G$ is a compact connected Lie group, $P\buildrel p\over\to X$ is a principal $G$-bundle with a fixed principal connection form $\Theta$, $\mathcal F\buildrel \rho \over \to X$ is the associated fiber bundle with fiber $F$ and the associated $G$-connection. We further assume that $F$ is a hamiltonian $G$-space with ($G$-equivariant) moment map $\Phi$.

Denote by $\mathcal F_\sharp$ and $\mathcal F^\sharp$ the manifolds defined in the pullback diagrams
\begin{eqnarray}
\begin{tikzcd}
\mathcal F_\sharp \arrow{r}{\widetilde{\tau_X}}\arrow{d}{\rho_\sharp} &\mathcal F \arrow{d}{\rho}\\
TX\arrow{r}{\tau_X} & X,
\end{tikzcd} \quad
\begin{tikzcd}
\mathcal F^\sharp \arrow{r}{\widetilde{\pi_X}}\arrow{d}{\rho^\sharp} &\mathcal F \arrow{d}{\rho}\\
T^*X\arrow{r}{\pi_X} & X
\end{tikzcd}
 \end{eqnarray} respectively. We note that the dual vector bundle of $\mathcal F_\sharp \buildrel {\widetilde{\tau_X}}\over \to \mathcal F$ is $\mathcal F^\sharp \buildrel {\widetilde{\pi_X}}\over \to \mathcal F$, so we have the canonical isomorphism
 $$
 \kappa: T^*\mathcal F^\sharp\to T^*\mathcal F_\sharp
 $$ as demonstrated in Section \ref{canonical isomorphism}. In view of the fact that Sternberg phase space $\mathcal F^\sharp$ is a symplectic manifold, we have a generalized Tulczyjew triple:
 \begin{eqnarray}\label{MTulczyjewTriple}
\begin{tikzcd}[column sep=small]
& T\mathcal F^\sharp \arrow{dr}{\alpha_{\mathscr F}} \arrow{dl}[swap]{\beta_{\mathscr F}}&\\
T^*\mathcal F^\sharp \arrow{rr}[swap]{\kappa}  & &  T^*\mathcal F_\sharp
\end{tikzcd}
\end{eqnarray} where isomorphism $\beta_{\mathscr F}$ comes from the symplectic structure on $\mathcal F^\sharp$ and $\alpha_{\mathscr F}:=\kappa\circ \beta_{\mathscr F}$. This generalized Tulczyjew triple shall be referred to as magnetized Tulczyjew triple.

Let
\begin{eqnarray}
 \fbox{$\vartheta^H_{\mathscr F}:=\beta_{\mathscr F}^*\vartheta_{\mathcal F^\sharp}$}, \quad 
 \fbox{$\vartheta^L_{\mathscr F}:=\alpha_{\mathscr F}^*\vartheta_{\mathcal F_\sharp}$} 
 \end{eqnarray}
 and $\hat \vartheta_X$ also denote the pullback of $\hat\vartheta_X$ under map $T\rho^\sharp$: $T\mathcal F^\sharp\to TT^*X$. Later we shall show that 
 \begin{eqnarray}\label{H=L}
\fbox{$\vartheta^L_{\mathscr F}-\vartheta^H_{\mathscr F}=\mathrm{d}\hat \vartheta_X$.}
\end{eqnarray}
Then $\mathrm{d}\vartheta^L_{\mathscr F}=\mathrm{d}\vartheta^H_{\mathscr F}=:\Omega_{\mathscr F}$, so $\vartheta^L_{\mathscr F}$ and $\vartheta^H_{\mathscr F}$ yield the same symplectic form $\Omega_{\mathscr F}$ on $T\mathcal F^\sharp$. We shall call $\vartheta^L_{\mathscr F}$ (resp. $\vartheta^H_{\mathscr F}$)  the {\bf Lagrangian} (resp. {\bf Hamiltonian}) {\bf Liouville form} on $T\mathcal F^\sharp$.

\subsubsection{Local formulae}
To get a local formula for $\beta_{\mathscr F}$, we need to get a local formula for the symplectic form on $\mathcal F^\sharp$. Since we work locally we may assume that $P\to X$ is $X\times G\to X$ and then the connection form
$\Theta$ is equal to $g^{-1}Ag +g^{-1}\, dg$ where $A$ is a $\mathfrak g$-valued differential one-form on $X$ and $g^{-1}\,dg$ is the Maurer-Cartan form on $G$. The Sternberg symplectic form $\omega_\Theta$ on $T^*X\times F$
is
$$
\omega_X+\Omega-{\mathrm d}\langle A, \Phi\rangle
$$
which, in local coordinates, is
$$
{\mathrm d}p_i\wedge {\mathrm d}q^i+{1\over 2}\Omega_{\alpha\beta}\, \mathrm{d}z^\alpha\wedge \mathrm{d}z^\beta-{1\over 2}\langle \partial_iA_j-\partial_jA_i, \Phi\rangle {\mathrm d}q^i\wedge {\mathrm d}q^j +\langle A_i, \partial_\alpha\Phi\rangle {\mathrm d}q^i\wedge \mathrm{d}z^\alpha.
$$
Then \underline{$\beta_{\mathscr F}$: $T\mathcal F^\sharp\to T^*\mathcal F^\sharp$ can be represented as}
 \begin{eqnarray}
(q, p, z, \dot q, \dot p, \dot z)&\mapsto& (q, p, z,\; \dot p_i-\langle \dot q^j(\partial_jA_i-\partial_iA_j), \Phi\rangle-\langle A_i, \dot z^\alpha\partial_\alpha\Phi\rangle,\cr
&& -\dot q,\; \dot z^\alpha \Omega_{\alpha\beta}+\langle \dot q^i A_i, \partial_\beta\Phi\rangle).
 \end{eqnarray}
Consequently \underline{$\alpha_{\mathscr F}$: $T\mathcal F^\sharp\to T^*\mathcal F_\sharp$ can be represented as}
  \begin{eqnarray}\label{local: alpha}
(q, p, z, \dot q, \dot p, \dot z)&\mapsto& (q, \dot q, z,\; \dot p_i-\langle \dot q^j(\partial_jA_i-\partial_iA_j), \Phi\rangle-\langle A_i, \dot z^\alpha\partial_\alpha\Phi\rangle,\cr
&& p_j,\; \dot z^\alpha \Omega_{\alpha\beta}+\langle \dot q^i A_i, \partial_\beta\Phi\rangle).
 \end{eqnarray}
 Then locally \underline{$\vartheta^H_{\mathscr F}$ is equal to}
 \begin{eqnarray}
(\dot p_i-\langle \dot q^j(\partial_jA_i-\partial_iA_j), \Phi\rangle-\langle A_i, \dot z^\alpha\partial_\alpha\Phi\rangle)\, {\mathrm d}q^i -\dot q^i\, {\mathrm d}p_i+(\dot z^\alpha \Omega_{\alpha\beta}+\langle \dot q^i A_i, \partial_\beta\Phi\rangle)\, \mathrm{d}z^\beta.\nonumber
\end{eqnarray}
Similarly, locally \underline{$\vartheta^L_{\mathscr F}$ is equal to}
\begin{eqnarray}
(\dot p_i-\langle \dot q^j(\partial_jA_i-\partial_iA_j), \Phi\rangle-\langle A_i, \dot z^\alpha\partial_\alpha\Phi\rangle)\,{\mathrm d}q^i +p_i\, \mathrm{d}\dot q^i+(\dot z^\alpha \Omega_{\alpha\beta}+\langle \dot q^i A_i, \partial_\beta\Phi\rangle)\,\mathrm{d}z^\beta.\nonumber
\end{eqnarray}
Since locally $\hat\vartheta_X=p_i\dot q^i$, identity \eqref{H=L} is verified, as promised. 

In terms of operators $\mathrm{i}_T$ and $\mathrm{d}_T$ on page \pageref{tangent lift} for tangent bundle $T\mathcal F^\sharp\to \mathcal F^\sharp$, we have
\begin{eqnarray}
\fbox{$\begin{array}{lll}
\vartheta_{\mathscr F}^H &= & \mathrm{i}_T\omega_\Theta, \\
\\
 \vartheta_{\mathscr F}^L &=  & \mathrm{d}_T\vartheta_X+\mathrm{i}_T \Omega_\Theta, \\
 \\
  \Omega_{\mathscr F}& = &\mathrm{d}_T\omega_\Theta 
\end{array}$}
\end{eqnarray}
whose validity can be checked by local computations based on formulae on page \pageref{tangent lift}.  Note that $\vartheta_X$ is really the pullback of $\vartheta_X$ under $\mathcal F^\sharp\to T^*X$ and $\Omega_\Theta$ is really the pullback of $\Omega_\Theta$ under $\mathcal F^\sharp\to \mathcal F$.

\subsubsection{Special symplectic structures on $T\mathcal F^\sharp$}\label{SpecialSymp}
Since $\mathcal F^\sharp$ is a symplectic manifold, by Proposition \ref{A: HamitonianSSS}, we have a {\bf Hamiltonian special symplectic manifold}
\[
\fbox{$S^H_{\mathscr F}:=(T{\mathcal F}^\sharp, {\mathcal F}^\sharp, \tau_{\mathcal F^\sharp}, \vartheta^H_{\mathscr F})$}
\] as usual.

Since diagram \eqref{D: commutativeSQ} is commutative, diagram
\begin{eqnarray}
\begin{tikzcd}[column sep=small]
T\mathcal F^\sharp \arrow{rr}{\tau_{\mathcal F^\sharp}}\arrow{d}[swap]{T\rho^\sharp} & &  \mathcal F^\sharp  \arrow{rr}{\widetilde{\tau_X}}\arrow{d}{\rho^\sharp}&& \mathcal F\arrow{d}{\rho}\\
TT^*X\arrow{rr}{\tau_{T^*X}}\arrow{d}[swap]{T\pi_X}& & T^*X \arrow{rr}{\pi_X}&& X\\
TX \arrow{urrrr}[swap]{\tau_X}&& && 
\end{tikzcd}\nonumber
\end{eqnarray}
is commutative, so we have a smooth map $T_\mathcal F$: $T\mathcal F^\sharp\to \mathcal F_\sharp$. Locally $T_{\mathcal F}$ can be represented as follows:
\begin{eqnarray}
(q, p, z, \dot q, \dot p, \dot z)&\mapsto& (q, \dot q, z).
 \end{eqnarray}
This local formulae, together with local formula \eqref{local: alpha}, implies that that diagram
\begin{eqnarray}
\begin{tikzcd}[column sep=small]
T\mathcal F^\sharp \arrow{rr}{\alpha_{\mathscr F}}\arrow{d}[swap]{T_{\mathcal F}} & &  T^*\mathcal F_\sharp \arrow{d}{\pi_{\mathcal F_\sharp}}\\
\mathcal F_\sharp\arrow{rr}{1}& & \mathcal F_{\sharp}\nonumber
\end{tikzcd}
\end{eqnarray}  is commutative. So quadruple 
$$\fbox{$S^L_{\mathscr F}:=(T{\mathcal F}^\sharp, {\mathcal F}_\sharp, T_{\mathcal F}, \vartheta^L_{\mathscr F})$}
$$ is a special symplectic manifold --- the {\bf Lagrangian special symplectic manifold}.  

\vskip 10pt
In summary, the magnetized Tulczyjew triple \eqref{MTulczyjewTriple} yields two special symplectic structures on $T\mathcal F^\sharp$, the one with Lagrangian Liouville one-form $\vartheta^L_{\mathscr F}$ is called the {\bf Lagrangian special symplectic structure} and the one with Hamiltonian Liouville one-form $\vartheta^H_{\mathscr F}$ is called the {\bf Hamiltonian special symplectic structure}.

\section{Tulczyjew's approach for particles in gauge fields}\label{Tulczyjew's approach}
Tulczyjew's approach to classical mechanics has two sides: the Hamiltonian side and the Lagrangian side, and the two sides are unified under a common geometric setting: the dynamics of a particle with configuration space $X$ is determined by a Lagrangian submanifold $D$ of $TT^*X$, whose description by its Hamiltonian $H$: $T^*X\to \mathbb R$ (resp. its Lagrangian $L$: $TX\to\mathbb R$) yields the Hamilton (resp. the Euler-Lagrange) equation. 

The discussion of Tulczyjew's approach shall be divided into two parts: the Hamiltonian formulation and the Lagrangian formulation.

\subsection{The Hamiltonian formulation}
Part of the reasons that Tulczyjew's approach works is the fact that $T^*X$ is a symplectic manifold. When we replace $T^*X$ by a symplectic manifold $(M, \omega)$, the Hamiltonian side still survives because Proposition \ref{A: HamitonianSSS} says that $(TM, M, \tau_M, \vartheta)$ is special symplectic manifold in which Tulczyjew has introduced \cite[page 249]{Tulczyjew1974}
\begin{definition}[Tulczyjew's Hamiltonian system]
A {\em Hamiltonian system} in special symplectic manifold $$(TM, M, \tau_M, \vartheta)$$ is a Lagrangian submanifold $N$ of $(TM, \mathrm{d}\vartheta)$ such that conditions for the existence of the unique form $\gamma$ stated in Proposition \ref{P: gamma} are satisfied and $\gamma$ is exact. The submanifold $K:=\tau_M(N)$ is called the {\em Hamiltonian constraint}, and a function $H$ on $K$ such that $\gamma=-\mathrm{d}H$, is called a {\em Hamiltonian} of $N$. 
\end{definition}
In particular, since $\mathcal F^\sharp$ is a special symplectic manifold, so we obtain Tulczyjew's Hamiltonian formulation for particles in gauge fields. 

To see how this is related to the ordinary Hamiltonian formulation for the dynamics of particles without constraint (i.e., $K=M$), one starts with a parametrized smooth curve $c$ on $M$. By taking derivative, we get a smooth parametrized curve $c'$ on $TM$.  Now the Hamilton equation for $c$ is nothing but the statement that $c'$ is a smooth parametrized curve on the Lagrangian submanifold $N$ of $TM$.  To verify this we just need to work locally, so we may assume that $(M, \omega)=(T^*\mathbb R^n, \omega_{\mathbb R^n})$ (Darboux's theorem), then $c=\left(q, p\right)$, so $
c'=\left(q, p, q', p'\right)$. Since the Lagrangian submanifold $N$ can be locally described as the set of points
$$
\left(q, p, {\partial H\over \partial p}, -{\partial H\over \partial q}\right),
$$ so $c'$ is a a smooth parametrized curve on $N$ means precisely that the Hamilton equation
\begin{eqnarray}
q'={\partial H\over \partial p}=\{q, H\}, \quad p'=-{\partial H\over \partial q}=\{p, H\}
\end{eqnarray} is satisfied.

\subsection{The Lagrangian formulation}
In this subsection we shall assume that $G$ is a compact connected Lie group, $P\buildrel p\over\to X$ is a principal $G$-bundle with a fixed principal connection form $\Theta$, $\mathcal F\buildrel \rho \over \to X$ is the associated fiber bundle with fiber $F$ and the associated $G$-connection. We further assume that $F$ is a hamiltonian $G$-space with ($G$-equivariant) moment map $\Phi$, symplectic form $\Omega$. Recall that $S^L_{\mathscr F}:=(T{\mathcal F}^\sharp, {\mathcal F}_\sharp, T_{\mathcal F}, \vartheta^L_{\mathscr F})$ is a special symplectic manifold. Mimicking Tulczyjew \cite[page 251]{Tulczyjew1974}, we have the following definition.

\begin{Definition}[Main Definition]
A {\em Lagrangian system} in $S^L_{\mathscr F}$ is a Lagrangian submanifold $N$ of $(T\mathcal F^\sharp, \mathrm{d}\vartheta^L_{\mathscr F})$ such that conditions for the existence of the unique form $\gamma$ stated in Proposition \ref{P: gamma} are satisfied and $\gamma$ is exact. The submanifold $J:= T_{\mathcal F}(N)$ is called the {\em Lagrangian constraint}, and a function $L$ on $J$ such that $\gamma=\mathrm{d}L$, is called a {\em Lagrangian} of $N$. 
\end{Definition}
Note that when $G$ is trivial and $F$ is a point, we have $S^L_{\mathscr F}=S^L_X$, then the above definition becomes Tulczyjew's definition 4.4 in Ref. \cite{Tulczyjew1974}.

\vskip 10pt
To get the Euler-Lagrange equation for the dynamics without constraint (i.e., $J=\mathcal F_\sharp$) , one starts with a parametrized smooth curve $\gamma$ on $\mathcal F$, then we get a parametrized smooth curve $(\gamma,(\rho\circ\gamma)')$ on $\mathcal F_\sharp$, so we arrive at a parametrized smooth curve $c$ on $T^*\mathcal F$ defined by the following commutative diagram
\begin{eqnarray}
\begin{tikzcd}[column sep=small]
& &  & & \mathbb R \arrow[bend right]{dllll}[description]{\gamma}\arrow{dll}[description]{(\gamma,(\rho\circ\gamma)')}\arrow[bend left, dashed]{drrrr}[description]{c}\arrow[dotted]{drr}[description]{c'}& & & & \\
\mathcal F && \mathcal F_\sharp\arrow{ll}[swap]{\widetilde{\tau_X}} \arrow{rr}{\mathrm{d}L}\arrow[bend right, dashed]{rrrrrr}[description]{\mathrm{Leg}_L}& &  T^*\mathcal F_\sharp \arrow{rr}[swap]{\alpha_{\mathscr F}^{-1}}{\cong}&& T\mathcal F^\sharp \arrow{rr}{\tau_{\mathcal F^\sharp }} && 
\mathcal F^\sharp.
\end{tikzcd}\nonumber
\end{eqnarray} 
By taking derivative, we get a smooth parametrized curve $c'$ on $T\mathcal F^\sharp $.  Now the Euler-Lagrange equation for $\gamma$ is nothing but the statement that $c'$ is  a smooth parametrized curve on the Lagrangian submanifold $N$ of $T\mathcal F^\sharp$. In terms of local coordinates, we can written $\gamma=(q,z)$. This really means that, locally we represent $\gamma$ by $(q, z)\circ \gamma$, but for notational sanity, $(q, z)\circ \gamma$ is also denoted by $(q, z)$. So $(q, z)$ is either a local function on $\mathcal F$ or a local function on $\mathbb R$, depending on the context. With this understood, we have $c=(q, {\partial L\over \partial \dot q}|_{\dot q = q'}, z)$, so
$$
c'=\left(q, \left.{\partial L\over \partial \dot q}\right|_{\dot q = q'}, z, q',  \left.{d\over dt}\left({\partial L\over \partial \dot q}\right|_{\dot q = q'}\right), z'\right). 
$$ 
Since the Lagrangian submanifold $N$ can be locally described as the set of points
$$
\left(q, {\partial L\over \partial \dot q}, z, \dot q, {\partial L\over \partial q^i}+\langle \dot q^j(\partial_jA_i-\partial_iA_j), \Phi \rangle+\langle A_i, \dot z^\beta\partial_\beta\Phi\rangle, \dot z^\alpha \right)
$$ where $\dot z^\alpha=\left(\partial_\beta L -\langle\dot q^iA_i, \partial_\beta \Phi\rangle\right)\Omega^{\beta\alpha}$ with $[\Omega^{\alpha\beta}]=[\Omega_{\alpha\beta}]^{-1}$.
Therefore, $c'$ is a a smooth parametrized curve on $N$ means that, locally
\begin{eqnarray}\label{local eqnofM}
\fbox{ \fontsize{7pt}{1em}\selectfont $
\begin{array}{rcl}
\displaystyle\frac{dz^\alpha}{dt} &=& \left(\displaystyle\left.{\partial L\over \partial z^\beta}\right|_{\dot q=q'} -\frac{dq^k}{dt}\left\langle A_k, \frac{\partial\Phi}{\partial z^\beta}\right\rangle\right)\Omega^{\beta\alpha}\\
\\
\displaystyle{d\over dt}\left(\left.{\partial L\over \partial \dot q^i}\right|_{\dot q= q'}\right) &= &\displaystyle\left.{\partial L\over \partial q^i}\right|_{\dot q=q'}+\frac{dq^j}{dt}\left\langle  \frac{\partial A_i}{\partial q^j}- \frac{\partial A_j}{\partial q^i}, \Phi\right\rangle+\frac{dz^\alpha}{dt}\left\langle A_i,  \frac{\partial\Phi}{\partial z^\alpha}\right\rangle
\end{array}
$}
\end{eqnarray}

\begin{Theorem}[Main Theorem]\label{main theorem} Assume the data in the beginning of this subsection and $L$:
$\mathcal F_\sharp\to \mathbb R$ is a smooth map.  For the unconstrained particle dynamics with configuration space $X$, internal space $F$, gauge field $\Theta$, and Lagrangian $L$, the following statements are true. 
\begin{enumerate}
\item The equation of motion can be locally written as Eq. \eqref{local eqnofM}. 

\item If $(x, y)$ is a Darboux (local) coordinate for $(F, \Omega)$ so that $\Omega=\mathrm{d}x\wedge \mathrm{d}y$ locally, then Eq. \eqref{local eqnofM} is the Euler-Lagrange equation for Lagrangian
\begin{eqnarray}\label{mLagrangian}
\mathcal L:= L-\langle\dot q^iA_i, \Phi\rangle+x\cdot\dot y.
\end{eqnarray}
\item The action functional for homologically trivial smooth loop $\gamma$: $\mathrm{S^1}\to \mathcal F$ is 
 \begin{eqnarray}\label{action}
S[\gamma]:=\displaystyle \int_{\mathrm{S}^1} (\gamma,(\rho\circ\gamma)')^*L\; dt+\int_{\Sigma} \tilde \gamma^*\Omega_\Theta
\end{eqnarray}
where $\Sigma$ is an oriented compact $2$-manifold with $\mathrm{S}^1=\partial \Sigma$, $\tilde \gamma$:  $\Sigma\to \mathcal F$ is a smooth extension of
$\gamma$, and $\Omega_\Theta$ is the Sternberg form on $\mathcal F$. 

\item There is a generalized version of Dirac's charge quantization condition:
 \begin{eqnarray}\label{quantization condition}
 \displaystyle \left[{\Omega_\Theta\over 2\pi }\right]\in \mathrm{Range}\left(H^2(\mathcal F; \mathbb Z)\buildrel\otimes_{\mathbb Z} 1\over \longrightarrow H^2(\mathcal F; \mathbb R)\right). 
 \end{eqnarray}
 I.e., the cohomology class represented by the closed real differential two-form ${\Omega_\Theta\over 2\pi }$ is an integral lattice point of the 2nd cohomology group of $\mathcal F$ with real coefficient. 
 
 \end{enumerate}
 \end{Theorem}
\begin{proof}
Lagrangian $\mathcal L$ in Eq. \eqref{mLagrangian} emerges out of a short straightforward computation in calculus of variations. Action $S$ in Eq. \eqref{action} is obtained from $\mathcal L$ with the help of Stokes' theorem in calculus and the fact that in dimension one the oriented bordism group coincides with the integral homology group. To arrive at the charge quantization condition \eqref{quantization condition} we demand the uniqueness of $\exp[iS[\gamma]]$, and use the fact that in dimension two the oriented bordism group coincides with the integral homology group as well as the universal coefficient theorem in algebraic topology. 
\end{proof}
It is equally easy to write down the equation of motion locally for a Lagrangian system with a generic Lagrangian constraint. 
\begin{remark}
The generalized Dirac quantization condition \eqref{quantization condition} is equivalent to the condition that the Sternberg phase is prequantizable, i.e., $[\omega_\Theta]\in H^2(\mathcal F^\sharp; \mathbb R)$ is integral. That is because $\mathcal F^\sharp$ and $\mathcal F$ are homotopy equivalent, and $\omega_X$ is an exact differential 2-form.
\end{remark}
\begin{remark}
The main theorem in particular tells us how to incorporate a Yang-Mills field in the Lagrangian formulation: just shift the ordinary Lagrangian $L(q, \dot q)$ by two terms: the Lorentz term $-\langle\dot q^iA_i(q), \Phi(y,x)\rangle$ and the new extra term $x\cdot\dot y$.  Note that $x\cdot\dot y$ can be replaced by $-y\cdot\dot x$ or ${1\over 2} (x\cdot\dot y-\dot y\cdot x)$, etc.
\end{remark}
For simplicity we shall write ${d\over dt}\left(\left.{\partial L\over \partial \dot q^i}\right|_{\dot q= q'}\right) $ as ${d\over dt}\left({\partial L\over \partial \dot q^i}\right)$, $\left.{\partial L\over \partial q^i}\right|_{\dot q=q'}$ as ${\partial L\over \partial q^i}$. Let
\begin{eqnarray}
\begin{array}{ll}
\partial_FL :=  \frac{\partial L}{\partial z^\alpha}\; \mathrm{d}z^\alpha, &\{f, g\}_F:=\Omega^{\alpha\beta}\frac{\partial f}{\partial z^\alpha}\frac{\partial g}{\partial z^\beta}\\
\\
 \frac{Dz}{dt} :=\frac{dz^\alpha}{dt}\frac{\partial}{\partial z^\alpha}+Y_{q'  \cdot A}, & F_{ji} :=\frac{\partial A_i}{\partial q^j}- \frac{\partial A_j}{\partial q^i}+[A_i, A_j].
 \end{array}
 \nonumber
\end{eqnarray}
With the help of Eqs. \eqref{phi1} and \eqref{phi2}, we have
\begin{corollary}
 Equation \eqref{local eqnofM} becomes
\begin{eqnarray} \label{eqnofM}
\fbox{$
\begin{array}{rcl}
\displaystyle{d\over dt}\left({\partial L\over \partial \dot q^i}\right) &= &\displaystyle{\partial L\over \partial q^i}+\frac{dq^j}{dt}\langle  F_{ji}, \Phi\rangle+\left\{L, \langle A_i, \Phi\rangle \right\}_F\\
\\
\displaystyle\frac{Dz}{dt}\lrcorner\, \Omega &=& \partial_FL
\end{array}
$}
\end{eqnarray}
provided that $F$ is a \underline{homogeneous} Hamiltonian $G$-space.
\end{corollary}

\begin{example}[Dynamics of electrically charged particles]
When $G=\mathrm{U}(1)$, $F$ is a co-adjoint orbit of $G$ (hence a point $-q_{\mathrm e}\in \mathbb R$) with $\Phi$ being the inclusion map, the second equation becomes $0=0$, so the equation of motion becomes the more familiar equation 
\begin{eqnarray}\label{abeliancase}
\displaystyle{d\over dt}\left({\partial L\over \partial \dot q}\right)  =\displaystyle {\partial L\over \partial q}-q_{\mathrm e} \; q'\lrcorner F.
\end{eqnarray} In particular, if $X=\mathbb R^3$, $L={1\over 2}\dot {\mathbf r}\cdot \dot {\mathbf r}$ and $\mathbf B=F_{23}\,\mathbf i+F_{31}\,\mathbf j+ F_{12}\,\mathbf k$, Eq. \eqref{abeliancase} becomes
the textbook equation of motion 
$$
\mathbf r'' =q_{\mathrm e}\; \mathbf r'\times \mathbf B
$$ with $m=1$ and $c=1$.
In case $X=\mathbb R^3 \setminus \{\mathbf 0\}$, and $\mathbf B=q_{\mathrm m}{\mathbf r\over r^3}$, condition \eqref{quantization condition} becomes Dirac's  {\em charge quantization condition} \cite{Dirac1931}: 
$$\displaystyle q_{\mathrm e}\,q_{\mathrm m}\in {1\over 2}\mathbb Z$$
with $\hbar =1$ and $c=1$.
\end{example}

\begin{example}[Wong's equations \cite{Wong70}]\label{Wong's equations} With the data in this subsection, we further assume that $X$ is a Lorentzian manifold with Lorentz metric $g$, $L={1\over 2}g_{\mu\nu}\dot q^\mu\dot q^\nu$, $F$ is a co-adjoint orbit of $G$ and $\Phi$ is the inclusion map. Let $p_\mu=g_{\mu\nu}\dot q^\nu$, $\Phi=-\xi_a \hat T^a$ where $\hat T^a$'s form a basis for $\frak g^*$, then the two equations in Eq. \eqref{eqnofM} become Wong's equations \cite[equations (1a) and (1b)]{Montgomery84} constrained to the Sternberg phase space $\mathcal F^\sharp$. Since $\mathcal F^\sharp$ is a symplectic leaf of Wong's phase space (a Poisson manifold) and solutions to Wong's equations are always constrained to symplectic leaves of Wong's phase space \cite[Theorem 2]{Montgomery84},  we effectively arrive at Wong's equations on Wong's phase space, for which a Lagrangian was constructed in Ref. \cite[equation (2.6)] {Balachandran1977}.

\end{example}

\begin{example}[Magnetized Kepler problems \cite{meng2013}] Let $k\ge 1$ be an integer and $\mu$ be a real number, we consider the $(2k+1)$-dimensional magnetized Kepler problem with magnetic charge $\mu$. Here, $X=\mathbb R^{2k+1}\setminus\{\mathbf 0\}$, $G=\mathrm{SO}(2k)$, 
$P\to X$ is the pullback of the principal $G$-bundle of $\mathrm{SO}(2k+1)\to \mathrm{S}^{2k}$ under map $\mathbf r\mapsto \frac{\mathbf r}{|\mathbf r|}$, $\Theta$ is the pullback of the canonical invariant connection on $\mathrm{SO}(2k+1)\to \mathrm{S}^{2k}$, $F$ is the co-adjoint orbit $\mathcal O_\mu$ of $G$ and $\Phi$ is just the inclusion map; see Ref. \cite{meng2013} for more details. The Hamiltonian is 
$$
H={1\over 2}\mathbf p\cdot \mathbf p-{1\over r}+{\mu^2/2k\over r^2},
$$ the pullback of a real function on $T^*X$, and the Lagrangian is
$$
L={1\over 2}\dot {\mathbf r}\cdot \dot {\mathbf r}+{1\over r}-{\mu^2/2k\over r^2},
$$ the pullback of a real function on $TX$. It is clear that $H$ and $L$ are related by the Legendre transformation. If $k>1$, the generalized Dirac quantization condition \eqref{quantization condition} is equivalent to the condition that the co-adjoint orbit $\mathcal O_\mu$ is prequantizable.
\end{example}

\begin{remark}
In the geometric approach here the Lagrangian takes simpler form $L$ and the equation of motion takes complicated form. In the textbook's approach to dynamics of electrically charged particles, the Lagrangian takes complicated form $\mathcal L$ and the equation of motion is still the simple-looking Euler-Lagrange equation. A similar remark is valid for the Hamiltonian approaches. ( Sternberg's Hamiltonian approach is our geometric approach to unconstrained Hamiltonian systems.) As usual, the geometric approach captures the essence of the problem, hence works much more generally.
\end{remark}
\begin{remark}
The dynamics of a charged particle is a Lagrangian submanifold of the symplectic manifold $T\mathcal F^\sharp$, and the two distinct special symplectic structures on $T\mathcal F^\sharp$ correspond to the two formulations of dynamics, one is Hamiltonian and one is Lagrangian. Since a Lagrangian submanifold of the symplectic manifold $T\mathcal F^\sharp$ may not come from a real function on either $\mathcal F^\sharp$ or $\mathcal F_\sharp$, this geometric approach goes much beyond what Sternberg's Hamiltonian approach can handle.
\end{remark}

Let us conclude this subsection with the following diagram:
$$
\begin{tikzcd}[row sep=scriptsize, column sep=scriptsize]
T^*_N\mathcal F^\sharp\arrow[hook]{r}{}\arrow[two heads]{d} &   T^*\mathcal F^\sharp  \arrow{dr}[description]{\pi_{\mathcal F^\sharp}}& & T\mathcal F^\sharp  \arrow{ll}[description]{\beta_{\mathscr F}} \arrow{rr}[description]{\alpha_{\mathscr F}}\arrow{dl}[description]{\tau_{\mathcal F^\sharp}}\arrow{dr}[description]{T_{\mathcal F}}&& T^*\mathcal F_\sharp\arrow{dl}[description]{\pi_{\mathcal F_\sharp}} &T^*_J \mathcal F_\sharp\arrow[hook]{l}\arrow[two heads]{d}\\
T^*N \arrow{r}{\pi_N}& N \arrow[hook]{r}\arrow{d}{H}\arrow[bend left]{l}[description]{-\mathrm{d}H}& \mathcal F^\sharp\arrow{dr}[description]{\widetilde{\pi^*}} & & \mathcal F_\sharp\arrow{dl}[description]{\tilde \pi} & J \arrow[hook]{l}\arrow{d}[swap]{L}\arrow[bend right]{r}[description]{\mathrm{d}L}& T^*J\arrow{l}[swap]{\pi_J}\\
&\mathbb R&&\mathcal F&&\mathbb R&\\
&&&\mathbb R\arrow{u}{\gamma}\arrow{uur}[swap]{(\gamma, (\rho\circ\gamma)')}&&&
\end{tikzcd}
$$

\subsection{The Legendre transformation}
For completeness we give a sketch of the Legendre transformation for the experts. The general readers who wish to know more details should consult sections 5 and 6 in Ref. \cite{Tulczyjew1974}.
\begin{definition}[The Legendre transformation]
The identity symplectic diffeomorphism $1_{T\mathcal F^\sharp}$ from $S_{\mathscr F}^L$ to $S_{\mathscr F}^H$ is called the Legendre transformation, and the identity symplectic diffeomorphism $1_{T\mathcal F^\sharp}$ from $S_{\mathscr F}^H$ to $S_{\mathscr F}^L$ is called the inverse Legendre transformation.
\end{definition}
Since $\vartheta_{\mathscr F}^L-\vartheta_{\mathscr F}^H=\mathrm{d}\hat \vartheta_X$, we have 
\begin{proposition} The Legendre transformation is generated by $-\varPhi_{\mathscr F}$ where 
$$\varPhi_{\mathscr F}:\quad \mathcal F_\sharp\times_{\mathcal F}\mathcal F^\sharp\to \mathbb R$$ is the map that sends $((q, \dot q, z), (q, p, z) )$ to $\langle \dot q, p\rangle$. The inverse Legendre transformation is generated by 
$$\tilde\varPhi_{\mathscr F}:\quad \mathcal F^\sharp\times_{\mathcal F}\mathcal F_\sharp\to \mathbb R$$ is the map that sends $((q, p, z), (q, \dot q, z) )$ to $\langle \dot q, p\rangle$.
\end{proposition}

If $f$ is a function of $x$ with a unique critical point $x_0$, we use $\mathrm{Stat}_x [f(x)]$ to denote $f(x_0)$.
\begin{example}
Let $N$ be a Lagrangian submanifold of $T\mathcal F^\sharp$ which is a generated by a map $L$:  $\mathcal F_\sharp\to \mathbb R$ and also a map $-H$: $\mathcal F^\sharp\to\mathbb R$. Then $L$ and $H$ are related by $H(y)=\mathrm{Stat}_x[\varPhi(x,y) - L(x)]$ subject to constraint $\widetilde{\tau_X}(x)=\widetilde{\pi_X}(y)$, and $L(x)=\mathrm{Stat}_y[\tilde\varPhi(y,x) - H(x)]$ subject to the same constraint. In terms of local coordinates, we have $H(q, p, z)=p\cdot \dot q - L(q, \dot q, z)$ evaluated at $\dot q$ such that $\frac{\partial L}{\partial \dot q^i}=p_i$, and $L(q, \dot q, z)=p\cdot \dot q - H(q, p, z)$ evaluated at $p$ such that $\frac{\partial H}{\partial p_i}=q^i$.

\end{example}

\appendix

 \section{Useful facts}\label{A: fact}

\subsection{Tangent lift \cite{pidello1987derivation}}\label{tangent lift}
Let $X$ be a smooth manifold and $TX$ be its total tangent space. The tangent lift, denoted by $\mathrm{d}_T$, is a map the sends a differential $k$-form on $X$ to a differential $k$-form on $TX$. By definition,
\begin{eqnarray}
\mathrm{d}_T=\mathrm{d}\circ\mathrm{i}_T+\mathrm{i}_T\circ\mathrm{d}
\end{eqnarray}
where $\mathrm{i}_T$ is a map that sends a differential $(k+1)$-form on $X$ to a differential $k$-form on $TX$ as follows: for $\alpha\in \Omega^{k+1}(X)$, $\mathrm{i}_T(\alpha)\in\Omega^k(TX)$ is defined via equation
\begin{eqnarray}
\mathrm{i}_T(\alpha)|_v =\tilde v\lrcorner \left.\tau_X^*\alpha\right |_v \quad\mbox{for any $v\in TX$}.
\end{eqnarray} Here $\tilde v\in T_vTX$ is any horizontal lift of $v\in T_{\tau_X(v)}X$.
For simplicity, we use the same symbol for a differential form on $X$ and its pullback on $TX$. Then, in terms of local coordinate $q$ and $(q, \dot q)$ on $X$ and $TX$, we have
\begin{eqnarray}
\mathrm{i}_T\left(\alpha_{i_0\cdots i_k}\, \mathrm{d}q^{i_0}\wedge\cdots\wedge \mathrm{d}q^{i_k}\right) =\sum_{j=0}^k(-1)^j\dot q^{i_j}\alpha_{i_0\cdots i_k}\, \mathrm{d}q^{i_0}\wedge\cdots\widehat{\mathrm{d}q^{i_j}}\cdots\wedge \mathrm{d}q^{i_k}\nonumber
\end{eqnarray} where the hat on $\mathrm{d}q^{i_j}$ means that $\mathrm{d}q^{i_j}$ is missing.

\subsection{``Cotangent map" $T^*_f$}
Let $f$: $Y\to X$ be a smooth map. While there is a canonical morphism from the tangent bundle of $Y$ to the tangent bundle of $X$, there is no canonical morphism from the cotangent bundle of $Y$ to the cotangent bundle of $X$. However, if $f$: $Y\to X$ is a fiber bundle with a (Ehresmann) connection, then, for any $y\in Y$, if $H_y$ is the horizontal subspace of $T_yY$, we have $T_{f(y)}X\cong H_y \subset T_yY$, so we have a linear map  $T_y^*Y\to T_{f(y)}^*X$ which, upon being globalized, becomes the top arrow $T^*_f$ in the commutative square
\begin{eqnarray}
\begin{tikzcd}
T^*Y \arrow{r}{T^*_f}\arrow{d}[swap]{\pi_Y} &T^*X \arrow{d}{\pi_X}\\
Y\arrow{r}{f}& X\nonumber
\end{tikzcd} 
\end{eqnarray} which is a morphism from the cotangent bundle of $Y$ to the cotangent bundle of $X$. 
\subsubsection{Local formula for $T^*_f$}
Here we assume that $f$: $Y\to X$ is the fiber bundle with fiber $F$, associated to the principal $G$-bundle $P\buildrel p\over \to X$ with a principal connection, so  $f$: $Y\to X$ has an associated $G$-connection. To work out a local formula for $T^*_f$: $T^*Y\to T^*X$, we may assume that $X$ is diffeomorphic to $\mathbb R^n$, $P\buildrel p\over \to X$ is trivial, and $F$ is diffeomorphic to $\mathbb R^l$. Upon fixing a diffeomorphism $q$: $X\to \mathbb R^n$,  a diffeomorphism $z$: $F\to \mathbb R^l$, and a local trivialization
\begin{eqnarray}
\begin{tikzcd}[column sep=small]
P \arrow{dr}[swap]{p} & &  X\times G \arrow{dl}{p_1}\arrow{ll}[swap]{\phi}{\cong}\\
& X & \nonumber
\end{tikzcd}
\end{eqnarray}
we have a diffeomorphism $Y\cong \mathbb R^n\times\mathbb R^l$ which shall be denoted by $(q, z)$. We also have diffeomorphisms
$$
T^*Y\cong \mathbb R^n\times \mathbb R^l\times (\mathbb R^n\times \mathbb R^l)^* \cong (\mathbb R^n\times \mathbb R^l)^2 \quad\mbox{and} \quad T^*X\cong \mathbb R^n\times (\mathbb R^n)^*\cong  (\mathbb R^n)^2$$
and they shall be denoted by $(q, z, p, y)$ and $(q,p)$ respectively. Finally, the infinitesimal action of $\mathfrak g$ on $F$ assigns a vector field $Y_\xi$ on $F$ to each $\xi\in \mathfrak g$. Let us also use $Y_\xi$ to denote the coordinator vector of $Y_\xi$ with respect to the local tangent frame $\frac{\partial}{\partial z^\alpha}$.

Under the trivialization $\phi$, the principal connection on $P\buildrel p\over\to X$ is represented by a $\mathfrak{g}$-valued differential one-form $A\cdot \mathrm{d}q :=A_i\,\mathrm{d}q^i $ on $X$, so the horizontal tangent vector of $Y$ can be represented by $$(q, z,  \dot q, - Y_{\dot q\cdot A}).$$
Therefore, the map $T^*Y\buildrel T^*_f\over\to T^*X$ can be represented by
\begin{eqnarray}\label{localT^*_f}
\fbox{$(q, z, p, y)\mapsto \left (q^i, p_j- y\cdot Y_{A_j}\right)$.}
\end{eqnarray}

\newpage
\section{A list of symbols}\label{A: notations}
\begin{eqnarray}
\begin{array}{ll}
\lrcorner & \quad\mbox{the interior product of vectors with forms}\cr
\wedge & \quad\mbox{the wedge product of forms}\cr
{\mathrm d} & \quad\mbox{the exterior derivative operator}\cr
{\mathrm d}_T & \quad\mbox{the tangent lift operator}\cr
\tau_X:\, TX\to X &\quad\mbox{the tangent bundle projection}\cr
\pi_X:\, T^*X\to X &\quad\mbox{the cotangent bundle projection}\cr
\vartheta_X, \omega_X: ={\mathrm d}\vartheta_X &\quad \mbox{the Liouville form, the symplectic form on $T^*X$}\cr
\hat\vartheta_X: TT^*X\to \mathbb R &\quad\mbox{the map such that $\hat\vartheta_X|_{T_\alpha (T^*X)}=\vartheta_X|_\alpha$}\cr
f^* & \quad\mbox{the pullback on differential forms under smooth map $f$ } \cr
Tf & \quad\mbox{the tangent map of $f$} \cr
T_mf & \quad\mbox{the linearization of $f$ at point $m$} \cr
T^*_f & \quad\mbox{see appendix \ref{A: fact} for definition} \cr
q, (q, \dot q), (q, p)&\quad\mbox{local coordinate map on $X$, $TX$, $T^*X$}\cr
\pi: E\to X&\quad \mbox{a real vector bundle over $X$}\cr 
p_E: T^*E\to E^*&\quad \mbox{see the explanation right before diagram \eqref{doublevb}}\cr
(q, u), (q, \alpha), (q, \hat u) &\quad\mbox{local coordinate map on $E$, $E^*$, $E^{**}$}\cr
q(t) & \quad\mbox{a smooth parametrized curve on $X$}\cr
q'(t) & \quad\mbox{the derivative of $q(t)$, it is a smooth parametrized curve on $TX$}\cr
\iota:\, E\to E^{**} & \quad\mbox{the usual identification}\cr
G & \quad\mbox{a compact connected Lie group, viewed as a} \cr
 & \quad\mbox{Lie subgroup of a rotation group} \cr
\mathfrak g, \mathfrak g^* & \quad\mbox{the Lie algebra of $G$ and its dual} \cr
\xi & \quad\mbox{an element in $\mathfrak g$}\cr
\langle\, ,\,\rangle & \quad\mbox{the paring of elements in vector space $V$ with elements in $V^*$}\cr
\mathrm{Ad}_a & \quad\mbox{the adjoint action of $a\in G$ on $\mathfrak g$}\cr
P\to X & \quad\mbox{a principal $G$-bundle} \cr
\Theta & \quad\mbox{a $\mathfrak g$-valued differential one-form on $P$ that}\cr 
&\quad\mbox{defines a principal connection on $P\to X$} \cr
R_a & \quad\mbox{the right action on $P$ by $a\in G$}\cr
X_\xi & \quad\mbox{the vector field on $P$ which represents the}\cr
&\quad\mbox{infinitesimal right action on $P$ by $\xi\in \mathfrak g$}\cr
\phi: X\times G\to P & \quad\mbox{a local trivialization of $P\to X$}\cr
A_\phi\;\mbox{or simply}\; A & \quad\mbox{a $\mathfrak g$-valued differential one-form on $X$ which}\cr
&\quad\mbox{locally represents $\Theta$ under trivialization $\phi$}\cr
F & \quad\mbox{a hamiltonian $G$-space}\cr
\mathcal F:=P\times_G F& \quad\mbox{$P\times F$ quotient by equivalence relation $(p, f)\sim(R_{a^{-1}}(p), L_a(f))$}\cr
z, (z, y)&\quad \mbox{local coordinate map on $F$, $T^*F$}\cr
\Phi: F\to \mathfrak g^* & \quad\mbox{the $G$-equivariant moment map}\cr
L_a & \quad\mbox{the left action on $F$ by $a\in G$}\cr
Y_\xi & \quad\mbox{the vector field on $F$ which represents the infinitesimal left action}\cr
&\quad\mbox{on $F$ by $\xi\in \mathfrak g$. It also denotes the coordinate vector of $Y_\xi$}\cr
\mathcal F \buildrel\rho\over\to X & \quad\mbox{the associated fiber bundle with fiber $F$}\cr
\mathcal F^\sharp  \buildrel\rho^\sharp\over\to T^*X, \mathcal F_\sharp  \buildrel\rho_\sharp\over\to TX& \quad\mbox{the pullback bundle of $\mathcal F \buildrel\rho\over\to X$ under $\pi_X$, $\tau_X$}\cr
\mbox{one forms $\vartheta^L_{\mathscr F}$ and $\vartheta^H_{\mathscr F}$} &\quad \mbox{see Subsection \ref{SpecialSymp} for their definition}
\end{array}\nonumber
\end{eqnarray}

\end{document}